\documentclass[11pt,a4paper]{article}
\usepackage[T1]{fontenc}
\usepackage[utf8]{inputenc}
\usepackage{authblk}
\usepackage{amscd}
\usepackage{amsmath}
\usepackage{mathrsfs}
\usepackage{cases}
\usepackage{amssymb}
\usepackage{amsfonts}
\usepackage{amssymb,amsfonts}

\usepackage{amsmath,amssymb,amsthm,hyperref}

\title{On explicit formulae of LMOV invariants}
\author[a,b]{Shengmao Zhu\thanks{szhu@zju.edu.cn, shengmaozhu@126.com}}
\affil[a]{Department of Mathematics, Zhejiang International Studies
University}

\affil[b]{Center of Mathematical Sciences, Zhejiang University}

\newtheorem{thm}{Theorem}[section]

\newtheorem{lemma}[thm]{Lemma}

\newtheorem{corollary*}[thm]{Corollary*}
\newtheorem{proposition}[thm]{Proposition}
\newtheorem{proposition*}[thm]{Proposition*}
\newtheorem{conjecture}[thm]{Conjecture}
\theoremstyle{definition}

\newtheorem{remark}[thm]{Remark}

\newtheorem{definition}[thm]{Definition}

\newtheorem{defn-thm}[thm]{Definition-Theorem}

\date{} 

\begin{document}
  \maketitle
\begin{abstract}
We started a program to study the open string integrality invariants
(LMOV invariants) for toric Calabi-Yau 3-folds with Aganagic-Vafa
brane (AV-brane) several years ago. This paper is devoted to the
case of resolved conifold with one out AV-brane in any integer
framing $\tau$, which is the large $N$ duality of Chern-Simons
theory for a framed unknot with integer framing $\tau$ in $S^3$. By
using the methods from string dualities,  we compute several
explicit formulae of the corresponding LMOV invariants for this
special model, whose integrality properties have been proved in a
separated paper \cite{LZhu}.
\end{abstract}

\section{Introduction}
We are interested in integrality structures of topological strings
theory. Let $X$ be a Calabi-Yau 3-fold with symplectic form
$\omega$, according to the work of Gopakumar and Vafa \cite{GV1},
the closed string free energy $F^X$, which is the generating
function of Gromov-Witten invariants $K_{g,Q}$, has the following
structure:
\begin{align*}
F^X=\sum_{g\geq 0}g_s^{2g-2}\sum_{Q\neq 0}K_{g,Q}e^{-Q
\cdot\omega}=\sum_{g\geq 0, d\geq 1}\sum_{Q\neq
0}\frac{1}{d}N_{g,Q}\left(2\sin\frac{dg_s}{2}\right)^{2g-2}e^{-dQ\cdot
\omega}
\end{align*}
where $N_{g,Q}$ are integers and vanish for large $g$ or $Q$. When
$X$ is a toric Calabi-Yau 3-fold, the above Gopakumar-Vafa
conjecture was proved in \cite{Peng,Konishi}.

Then we are going to investigate the integrality structures of open
topological strings. Let us consider a Calabi-Yau 3-fold $X$ with a
Lagrangian submanifold $\mathcal{D}$ in it. Based on Ooguri and
Vafa's work \cite{OV}, the generating function of open Gromov-Witten
invariants can also be expressed in terms of a series of new
integers which were later refined by Labastida, Mari\~no and Vafa in
\cite{LM1,LM2,LMV}:
\begin{align} \label{openmultiplecoveringintro}
&\sum_{g\geq 0}\sum_{Q\neq 0}g_s^{2g-2+l(\mu)}K_{\mu,g,Q}e^{-Q\cdot
\omega}\\\nonumber &=\sum_{g\geq 0}\sum_{Q\neq
0}\sum_{d|\mu}\frac{(-1)^{l(\mu)+g}}{\prod_{i=1}^{l(\mu)}\mu_i}d^{l(\mu)-1}n_{\mu/d,g,Q}\prod_{j=1}^{l(\mu)}
(2\sin\frac{\mu_jg_s}{2})(2\sin\frac{dg_s}{2})^{2g-2}e^{-dQ\cdot
\omega}.
\end{align}
These new integers $n_{\mu,g,Q}$ (here $\mu$ denote a partition of a
positive integer) are referred as LMOV invariants in this paper.

Although for any toric Calabi-Yau 3-fold with Agangica-Vafa brane
(AV-brane for short) \cite{AV}, we have an effective  method  to
calculate the open string partition function by gluing topological
vertices \cite{AKMV,LLLZ}, it is difficult to compute the
corresponding LMOV invariants $n_{\mu,g,Q}$ and prove their
integrality properties.

During the past several years, we started a program to study the
LMOV invariants and their variations
\cite{CLPZ,CZhu,LZhu0,LZhu,Zhu3,Zhu4,Zhu5}. In this paper, we only
focus on a special toric Calabi-Yau 3-fold, i.e. the resolved
conifold $\hat{X}$ with one special Lagrangian submanifold (AV-brane
$D_\tau$ in integer framing $\tau$). More general cases will be
discussed in a separated paper \cite{Zhu5}. We also refer to
\cite{MMMRKS,MMMS,PS} and the references therein for some recent
developments.

Based on large $N$ duality, the open string theory of
$(\hat{X},D_\tau)$ is the large $N$ duality of the $SU(N)$
Chern-Simons theory of $(S^3,U_\tau)$, where $U_\tau$ denotes a
framed unknot (trivial knot) with integer framing $\tau$. The large
$N$ duality of Chern-Simons and topological string theory was
proposed by Witten \cite{W3}, and developed further by
\cite{GV2,OV,LMV}. Later, Mari\~no and Vafa \cite{MV} generalized it
to the case of a knot with arbitrary framing. The large $N$ duality
of $(\hat{X},D_\tau)$ and $(S^3,U_\tau)$ can be  expressed in terms
of the following identity:
\begin{align} \label{LargeNdualiyofunknot}
Z_{CS}^{(S^3,U_\tau)}(q,a;\mathbf{x})=Z_{str}^{(\hat{X},D_{\tau})}(g_s,a;\mathbf{x}),
\ q=e^{\sqrt{-1}g_s}
\end{align}
where the explicit expressions of the above two partition functions
in identity (\ref{LargeNdualiyofunknot}) are given by the formulae
(\ref{partitionfunctionframedunknot}) and
(\ref{partitonfunctionresolvedconifold}) respectively.  The identity
(\ref{LargeNdualiyofunknot}) implies the Mari\~no-Vafa formula
\cite{MV,LLZ,OP}, a very powerful Hodge integral identity, which
implies various important results in intersection theory of moduli
spaces of curves, see \cite{Zhu1,LiuZhu} for a review of the
applications of Mari\~no-Vafa formula. Finally, the identity
(\ref{LargeNdualiyofunknot}) was proved by J. Zhou \cite{Zhou} based
on his previous joint works with C.-C. Liu and K. Liu
\cite{LLZ,LLZ2}.

On the other hand side, through mirror symmetry, the partition
function $Z_{str}^{(\hat{X},D_{\tau})}(g_s,a;\mathbf{x})$ can also
be computed by topological string $B$-model \cite{BCOV}. The mirror
geometry information of $(\hat{X},D_{\tau})$ is encoded in a mirror
curve $\mathcal{C}_{\hat{X}}$. Then the disc counting invariants of
$(\hat{X},D_{\tau})$ were given by the coefficients of
superpotential related to the mirror curve \cite{AV,AKV}, and this
fact was proved in \cite{FL}. Furthermore, the open Gromov-Witten
invariants of higher genus with more holes can be obtained by using
Eynard-Orantin topological recursions \cite{EO1}. This approach
named BKMP conjecture, was proposed by Bouchard, Klemm, Mari\~no and
Pasquetti in \cite{BKMP}, and then fully proved in \cite{EO2,FLZ}
for any toric Calabi-Yau 3-fold with  AV-brane, so one can also use
the BKMP method to compute the LMOV invariants for
$(\hat{X},D_{\tau})$.

In conclusion, now we have three different approaches to compute the
open string partition function
$Z_{str}^{(\hat{X},D_{\tau})}(g_s,a;\mathbf{x})$: (i) topological
vertex \cite{AKMV,LLLZ}; (ii) Chern-Simons partition function
(\ref{partitionfunctionframedunknot}); (iii) BKMP method
\cite{BKMP}.

In this paper, we provide several explicit formulae for LMOV
invariants of the open string model $(\hat{X},D_{\tau})$ by using
above methods. Firstly, we illustrate the computations of the mirror
curve of $(\hat{X},D_{\tau})$ by the approach in \cite{AV2}. It
turns out that the mirror curve of this model is given by:
\begin{align} \label{mirrorcurve}
y-1-a^{-\frac{1}{2}}(-1)^\tau xy^{\tau}(ay-1)=0.
\end{align}
By using the mirror curve (\ref{mirrorcurve}), we obtain an explicit
formula for genus $0$ and one-hole LMOV invariants (disc countings)
$n_{m,0,l-\frac{m}{2}}(\tau)$ which is denoted by $n_{m,l}(\tau)$
for brevity:
\begin{align*}
n_{m,l}(\tau)=\sum_{d|m,d|l}\frac{\mu(d)}{d^2}c_{\frac{m}{d},\frac{l}{d}}(\tau),
\end{align*}
where
\begin{align*}
c_{m,l}(\tau)=-\frac{(-1)^{m\tau+m+l}}{m^2}\binom{m}{l}\binom{m\tau+l-1}{m-1}
\end{align*}
and $\mu(d)$ denotes the M\"{o}bius functions. In \cite{LZhu}, we
have proved that  $n_{m,l}(\tau) \in \mathbb{Z}$ for any $\tau\in
\mathbb{Z}$, $m\geq 1, l\geq 0$.

\begin{remark}
Recently, Panfil and Sulkowski \cite{PS} generalized the above disc
counting formula to a class of toric Calabi-Yau manifolds without
compact four-cycles which is also referred to as strip geometries.
In our notations, their formula (cf. formula (4.19) in \cite{PS})
can be formulated as follow. Given two integers $r,s\geq 0$, set
$\mathbf{l}=(l_1,...,l_r)$, $\mathbf{k}=(k_1,...,k_s)$, and
$|\mathbf{l}|=\sum_{j=1}^rl_j$, $|\mathbf{k}|=\sum_{j=1}^sk_j$. We
define
\begin{align*}
c_{m,\mathbf{l},\mathbf{k}}(\tau)=\frac{(-1)^{m(\tau+1)+|\mathbf{l}|}}{m^2}
\binom{m\tau+|\mathbf{l}|+|\mathbf{k}|-1}{m-1}\prod_{j=1}^{r}\binom{m}{l_j}\prod_{j=1}^s
\frac{m}{m+k_j}\binom{m+k_j}{k_j}.
\end{align*}
Then, we have the following disc counting formula
\begin{align*}
n_{m,\mathbf{l},\mathbf{k}}(\tau)
=\sum_{d|\text{gcd}(m,\mathbf{l},\mathbf{k})}\frac{\mu(d)}{d^2}c_{m/d,\mathbf{l}/d,\mathbf{k}/d}(\tau)
\end{align*}
In \cite{Zhu5}, we have generalized the number theory method used in
\cite{LZhu} to show that $n_{m,\mathbf{l},\mathbf{k}}(\tau)\in
\mathbb{Z}$.
\end{remark}

For the LMOV invariants of genus 0 with two holes, we study the
Bergmann kernel expansion in the BKMP construction, and find an
explicit formula for the LMOV invariants
$n_{(m_1,m_2),0,\frac{m_1+m_2}{2}}(\tau)$ which is denoted by
$n_{(m_1,m_2)}(\tau)$ for short,
\begin{align*}
    n_{(m_1,m_2)}(\tau)& =  \frac{1}{m_1+m_2}\sum_{d\mid m_1,d\mid m_2} \mu(d) (-1)^{(m_1+m_2)(\tau+1)/d} \nonumber \\
                     & \cdot
                     \binom{(m_1\tau+m_1)/d-1}{m_1/d}\binom{(m_2\tau+m_2)/d}{m_2/d}.
\end{align*}
In \cite{LZhu}, we have also proved that $n_{(m_1,m_2)}(\tau)\in
\mathbb{Z}$ for $m_1,m_2\geq 1$ and $\tau\in \mathbb{Z}$.

As to the genus $0$ LMOV invariants with more than two holes, one
can compute the LMOV invariant $n_{\mu,0,Q}(\tau)$ for general $Q$
by using the BKMP construction. But it is hard to give an explicit
formula for general $Q$, except $Q=\frac{|\mu|}{2}$ in which case
\begin{align*}
n_{\mu,0,\frac{|\mu|}{2}}(\tau)=(-1)^{l(\mu)}\sum_{d|\mu}\mu(d)d^{l(\mu)-1}K^{\tau}_{\frac{\mu}{d},0,\frac{|\mu|}{2d}}
\end{align*}
where
\begin{align*}
K^\tau_{\mu,0,\frac{|\mu|}{2}}=(-1)^{|\mu|\tau}[\tau(\tau+1)]^{l(\mu)-1}\prod_{i=1}^{l(\mu)}
\binom{\mu_i(\tau+1)-1}{\mu_i-1}\left(\sum_{i=1}^{l(\mu)}\mu_i\right)^{l(\mu)-3}.
\end{align*}
It is obvious that $K^\tau_{\mu,0,\frac{|\mu|}{2}}\in \mathbb{Z}$
for any $\tau\in \mathbb{Z}$, and since $l(\mu)\geq 3$, it
immediately implies that $n_{\mu,0,\frac{|\mu|}{2}}(\tau)\in
\mathbb{Z}$ for any partition $\mu$ with $l(\mu)\geq 3$ and $\tau
\in \mathbb{Z}$.

Finally, we study the high genus LMOV invariants $n_{m,g,Q}(\tau)$
of framed unknot $U_\tau$. We define
\begin{align*}
g_m(q,a)=\sum_{d|m}\mu(d)\mathcal{Z}_{m/d}(q^d,a^d),
\end{align*}
where
$\mathcal{Z}_{m}(q,a)=(-1)^{m\tau}\sum_{|\nu|=m}\frac{1}{\mathfrak{z}_\nu}
\frac{\{m\nu\tau\}}{\{m\}\{m\tau\}}\frac{\{\nu\}_a}{\{\nu\}}$. In
\cite{LZhu}, we have proved that $g_{m}(q,a)$ is a polynomial of
  $n_{m,g,Q}(\tau)$.
More precisely,
\begin{align*}
g_m(q,a)=\sum_{g\geq 0}\sum_{Q}n_{m,g,Q}(\tau)z^{2g-2}a^Q\in
z^{-2}\mathbb{Z}[z^2,a^{\pm \frac{1}{2}}],
\end{align*}
where $z=q^{\frac{1}{2}}-q^{-\frac{1}{2}}$. In other words, we have
\begin{align*}
n_{m,g,Q}(\tau)=\text{Coefficient of term $z^{2g-2}a^Q$ in the
polynomial } g_m(q,a).
\end{align*}

The rest of this paper is organized as follow: In Section 2, we
review the mathematical definitions of topological string partition
functions, free energies, and the integrality structures appearing
in topological strings. Then we introduce the LMOV invariants in
open topological strings. In Section 3, we first review Witten's
Chern-Simons theory for three-manifolds and links, and the large $N$
duality between the Chern-Simons theory and topological strings.
Then, the basic case for framed unknot was illustrated explicitly.
We also formulate the LMOV integrality conjecture for framed knot.
In Section 4, we study the LMOV invariants for framed unknot in
detail. We first illustrate the computations of the mirror curve of
$(\hat{X},D_{\tau})$ by using the approach of \cite{AV2}. Then we
compute the explicit formulae for genus 0 LMOV invariants by using
this mirror curve, and high genus LMOV invariants with one hole as
well. Finally, in Section 5, we discuss several related questions
and works.

\textbf{Acknowledgements.} This paper is grew out of our preprint
\cite{LZhu0} on arXiv,  some part of \cite{LZhu0} has been published
in \cite{LZhu}, the author would like to thank Dr. Wei Luo for the
collaboration of \cite{LZhu}.

\section{Topological strings}
\subsection{Closed topological strings and Gromov-Witten invariants}
Topological strings on a Calabi-Yau 3-fold $X$ have two types:
A-models and B-models. The mathematical theory for A-model is
Gromov-Witten theory \cite{Ko,HKKPTVVZ}. Let
$\overline{\mathcal{M}}_{g,n}(X,Q)$ be the moduli space of stable
maps $(f, \Sigma_g,p_1,..,p_n)$, where $f: \Sigma_g\rightarrow X$ is
a holomorphic map from the nodal curve $\Sigma_g$ to the K\"{a}hler
manifold $X$ with $f_*([\Sigma_g])=Q\in H_2(X,\mathbb{Z})$. In
general, $\overline{\mathcal{M}}_{g,n}(X,Q)$ carries a virtual
fundamental class $[\overline{\mathcal{M}}_{g,n}(X,Q)]^{vir}$ in the
sense of \cite{BF,LT}. The virtual dimension
 is given by:
\begin{align*}
\text{vdim} [\overline{\mathcal{M}}_{g,n}(X,Q)]^{vir}=\int_Q
c_1(X)+(\dim X-3)(1-g)+n.
\end{align*}
When $X$ is a Calabi-Yau 3-fold, i.e. $c_1(X)=0$, then vdim$
[\overline{\mathcal{M}}_{g,0}(X,Q)]^{vir}=0$. The genus $g$, degree
$Q$ Gromov-Witten invariants of $X$ is defined by
\begin{align*}
K_{g,Q}=\int_{[\overline{\mathcal{M}}_{g,0}(X,Q)]^{vir}}1
\end{align*}
which is usually denoted by $K_{g,Q}$ for brevity without any
confusions. In  A-model, the genus $g$ closed free energy $F_{g}^X$
of $X$ is the generating function of Gromov-Witten invariants
$K_{g,Q}$, i.e.
\begin{align*}
F_{g}^X=\sum_{Q\neq 0}K_{g,Q}e^{-Q \cdot\omega},
\end{align*}
where $\omega$ represents the K\"ahler class for $X$. We define the
total free energy $F^X$ and partition function $Z^X$ as
\begin{align*}
F^X=\sum_{g\geq 0}g_s^{2g-2}F_{g}^X, \ Z^X=\exp(F^X),
\end{align*}
where $g_s$ denotes the string coupling constant. In mathematics,
 the free energy $F^{X}$ are mainly computed by the method of
localizations \cite{Ko,GP}. Especially, when $X$ is a toric
Calabi-Yau 3-fold, we have a more effective approach to obtain the
partition function $Z^X$ by the method of gluing topological
vertices \cite{AKMV,LLLZ}.

 Usually, the Gromov-Witten invariants
$K_{g,Q}$ are rational numbers, from the BPS counting in M-theory,
Gopakumar and Vafa \cite{GV1} expressed the total free energy $F^X$
in terms of the generating function of a series of new integer
numbers $N_{g,Q}$ as follow:
\begin{align*}
F^X=\sum_{g\geq 0}g_s^{2g-2}\sum_{Q\neq 0}K_{g,Q}e^{-Q
\cdot\omega}=\sum_{g\geq 0, d\geq 1}\sum_{Q\neq
0}\frac{1}{d}N_{g,Q}\left(2\sin\frac{dg_s}{2}\right)^{2g-2}e^{-dQ\cdot
\omega}
\end{align*}
The integrality of Gopakumar-Vafa invariants $N_{g,Q}$ was first
proved by P. Peng for the case of toric Del Pezzo surfaces
\cite{Peng}. The proof for general toric Calabi-Yau threefolds was
then given by Konishi in \cite{Konishi}.

\subsection{Open topological strings}

Let us now consider the open sector of topological A-model of a
Calabi-Yau 3-fold $X$ with a Lagrangian submanifold $\mathcal{D}$
with dim $H_{1}(\mathcal{D},\mathbb{Z})=L$. The open sector
topological A-model can be described by holomorphic maps $\phi$ from
open Riemann surface of genus $g$ and $l$-holes $\Sigma_{g,l}$ to
$X$, with Dirichlet condition specified by $\mathcal{D}$.  These
holomorphic maps are referred as open string instantons. More
precisely, an open string instanton is a holomorphic map $\phi:
\Sigma_{g,l}\rightarrow X$ such that $\partial \Sigma_{g,l}=\cup
_{i=1}^{l}\mathcal{C}_i\rightarrow \mathcal{D}\subset X$ where the
boundary $\partial \Sigma_{g,l}$ of $\Sigma_{g,l}$ consists of $l$
connected components $\mathcal{C}_i$ mapped to Lagrangian
submanifold $\mathcal{D}$ of $X$. Therefore, the open string
instanton $\phi$ is described by the following two different kinds
of data: the first is the ``bulk part" which is given  by
$\phi_*[\Sigma_{g,l}]=Q\in H_2(X,\mathcal{D})$, and the second is
the ``boundary part" which is given by
$\phi_*[\mathcal{C}_i]=w^\alpha_i\gamma_\alpha$, for $i=1,..l$,
where $\gamma_\alpha$,\ $\alpha=1,..,L$ is a basis of
$H_{1}(\mathcal{D},\mathbb{Z})$ and $w^\alpha_i\in \mathbb{Z}$. Let
$\vec{w}=(w^1,..,w^L)$, and where
$w^\alpha=(w^\alpha_1,...,w^\alpha_l)\in \mathbb{Z}^l$, for
$\alpha=1,...,L$. We expect there exist the corresponding open
Gromov-Witten invariants $K_{\vec{w},g,Q}$ determined by the data
 $\vec{w}, Q$ in the genus $g$. See \cite{LS,KL} for mathematical
aspects of defining these invariants in special cases.

We take all $w_i\geq 1$ as in \cite{MV}, and use the notations of
partitions and symmetric functions \cite{Mac}. We denote by
$\mathcal{P}$ the set of all partitions including the empty
partition $0$, and by $\mathcal{P}_+$ the set of nonzero partitions.
Let $\mathbf{x}=\{x_1,x_2,...\}$ be the set of infinitely many
independent variables. For $n\geq 0$, let
$p_n(\mathbf{x})=\sum_{i\geq 1}x_i^{n}$ be a power sum symmetric
function. For a partition $\mu\in \mathcal{P}_{+}$, set
$p_{\mu}(\mathbf{x})=\prod_{i=1}^{h}p_{\mu_i}(\mathbf{x})$. For
$\vec{\mu}\in \mathcal{P}^L$, and
$\vec{\mathbf{x}}=(\mathbf{x}^1,..,\mathbf{x}^L)$, let
$p_{\vec{\mu}}(\vec{\mathbf{x}})=\prod_{\alpha=1}^Lp_{\mu^\alpha}(\mathbf{x}^\alpha)$.
The total free energy and partition function of open topological
string on $(X,\mathcal{D})$ are expressed in the following forms:
\begin{align*}
F_{str}^{(X,\mathcal{D})}(g_s,\omega,\vec{\mathbf{x}})&=\sum_{g\geq
0}\sum_{\vec{\mu} \in
\mathcal{P}^L\setminus{\{0\}}}\frac{1}{|Aut(\vec{\mu})|}g_{s}^{2g-2+l(\mu)}
\sum_{Q\neq 0}K_{\vec{\mu},g,Q}e^{-Q\cdot
\omega}p_{\vec{\mu}}(\vec{\mathbf{x}})\\\nonumber
Z_{str}^{(X,\mathcal{D})}(g_s,\omega,\vec{\mathbf{x}})&=\exp(F_{str}^{(X,\mathcal{D})}(g_s,\omega,\vec{\mathbf{x}})).
\end{align*}

The central problem in open topological string theory is how to
calculate the partition function
$Z_{str}^{(X,\mathcal{D})}(g_s,\omega,\vec{\mathbf{x}})$ or the open
Gromov-Witten invariants $K_{\vec{\mu},g,Q}$. For the case of
compact Calabi-Yau 3-folds, such as the quintic $X_5$, there are
only a few works devoted to the study of its open Gromov-Witten
invariants, for example, a complete calculation of the disk
invariants of $X_5$ with boundary in a real Lagrangian was given in
\cite{PSW}.

Suppose $X$ is a toric Calabi-Yau 3-fold, and $\mathcal{D}$ is a
special Lagrangian submanifold named as Aganagic-Vafa A-brane in the
sense of \cite{AV,AKV}. The open string partition function
$Z_{str}^{(X,\mathcal{D})}(g_s,\omega,\mathbf{x})$ can be computed
by the method of topological vertex \cite{AKMV,LLLZ} and the method
of topological recursion developed by Eynard and Orantin \cite{EO1}.
The second approach was first proposed by Mari\~no \cite{Mar}, and
studied further by Bouchard, Klemm, Mari\~no and Pasquetti
\cite{BKMP}, the equivalence of the two methods was proved in
\cite{EO2,FLZ}

In the following, we only need to consider the case of $L$=1. It is
also useful to introduce the generating function of $K_{\mu,g,Q}$ in
the fixed genus $g$ as follow:
\begin{align*}
F_{(g,l)}^{(X,\mathcal{D})}=\sum_{\mu\in \mathcal{P}_+}\sum_{Q\neq
0}K_{\mu,g,Q}e^{-Q\cdot \omega}x_1^{\mu_1}\cdots x_l^{\mu_l}.
\end{align*}

\subsection{Integrality structures and LMOV invariants}
We introduce the new variables $q=e^{\sqrt{-1}g_s}$,
$a=e^{-\omega}$, and let $f_{\lambda}(q,a)$ be a function determined
by the following formula
\begin{align*}
F_{str}^{(X,\mathcal{D})}(g_s,a,\mathbf{x})=\sum_{d=1}^{\infty}\frac{1}{d}\sum_{\lambda\in
\mathcal{P}^+}f_{\lambda}(q^d,a^d)s_{\lambda}(\mathbf{x}^d),
\end{align*}
where $s_{\lambda}(\mathbf{x})$ is the Schur symmetric functions
\cite{Mac}.

 Just as in the closed string case \cite{GV1},  the open
topological strings compute the partition function of BPS domain
walls in a related superstring theory \cite{OV}. It follows that
$F^{(X,\mathcal{D})}$ also has an integral expansion. This
integrality structure was further refined in \cite{LM1,LM2,LMV}
which showed that $f_{\lambda}(q,a)$ has the following integral
expansion
\begin{align*}
f_{\lambda}(q,a)=\sum_{g=0}^{\infty}\sum_{Q\neq
0}\sum_{|\mu|=|\lambda|}M_{\lambda\mu}(q)
N_{\mu,g,Q}(q^{\frac{1}{2}}-q^{-\frac{1}{2}})^{2g-2}a^Q,
\end{align*}
where  $N_{\mu,g,Q}$ are integers which compute the net number of
BPS domain walls and $M_{\lambda\mu}(q)$ is defined by
\begin{align} \label{Mlambdamu}
M_{\lambda\mu}(q)=\sum_{\mu}\frac{\chi_{\lambda}(C_{\nu})\chi_{\mu}(C_{\nu})}{\frak{z}_{\nu}}\prod_{j=1}^{l(\nu)}(q^{-\nu_{j}/2}-q^{\nu_{j}/2}),
\end{align}
where $\chi_{\nu}(C_\mu)$ is the character of an irreducible
representation of the symmetric group and
$\mathfrak{z}_\mu=|Aut(\mu)|\prod_{i=1}^{l(\mu)}\mu_i$.

For convenience, we usually introduce the invariant
\begin{align} \label{formula-invariantchange}
n_{\mu,g,Q}=\sum_{\nu}\chi_{\nu}(C_\mu)N_{\nu,g,Q}.
\end{align}
\begin{definition}
These predicted integers $N_{\mu,g,Q}$ and  $n_{\mu,g,Q}$ are both
called LMOV invariants.
\end{definition}
Therefore,
\begin{align*}
f_{\lambda}(q,a)=\sum_{g\geq 0}\sum_{Q\neq 0}\sum_{\mu \in
\mathcal{P}}\frac{\chi_{\lambda}(C_{\mu})}{\mathfrak{z}_{\mu}}n_{\mu,g,Q}\prod_{j=1}^{l(\mu)}
(q^{-\frac{\mu_j}{2}}-q^{\frac{\mu_j}{2}})(q^{-\frac{1}{2}}-q^{\frac{1}{2}})^{2g-2}a^Q
\end{align*}
By applying the orthogonal relation $
\sum_{\lambda}\frac{\chi_{\lambda}(C_{\mu})\chi_{\lambda}(C_{\nu})}{\mathfrak{z}_{\mu}}=\delta_{\mu,\nu},
$ we obtain the following multiple covering formula for open
topological string:
\begin{align} \label{openmultiplecovering}
&\sum_{g\geq 0}\sum_{Q\neq
0}g_s^{2g-2+l(\mu)}K_{\mu,g,Q}a^Q\\\nonumber &=\sum_{g\geq
0}\sum_{Q\neq
0}\sum_{d|\mu}\frac{(-1)^{l(\mu)+g}}{\prod_{i=1}^{l(\mu)}\mu_i}d^{l(\mu)-1}n_{\mu/d,g,Q}\prod_{j=1}^{l(\mu)}
(2\sin\frac{\mu_jg_s}{2})(2\sin\frac{dg_s}{2})^{2g-2}a^{dQ}.
\end{align}
Hence we have the following integrality structure conjecture which
is referred as the Labastida-Mari\~no-Ooguri-Vafa (LMOV) conjecture
for open topological string.
\begin{conjecture}[LMOV conjecture for open topological string]
\label{LMOVforopenstring} Let $F_{\mu}^{(X,\mathcal{D})}$ be the
generating function defined by
\begin{align*}
F_{str}^{(X,\mathcal{D})}(g_s,a,\mathbf{x})=\sum_{\mu\in
\mathcal{P}_+}F_{\mu}^{(X,\mathcal{D})}p_{\mu}(\mathbf{x}),
\end{align*}
then $F_{\mu}^{(X,\mathcal{D})}$ has the integral expansion as in
the righthand side of the formula (\ref{openmultiplecovering}).
\end{conjecture}
There is no general definition for the open Gromov-Witten invariants
$K_{\mu,g,Q}$. However, just as mentioned in the previous
subsection, when $X$ is a toric Calabi-Yau 3-fold, and $\mathcal{D}$
is  the Aganagic-Vafa A-brane \cite{AV}, the open string partition
function $Z_{str}^{(X,\mathcal{D})}$ can be fully determined by
using the method of topological vertex \cite{AKMV,LLLZ}, and the
open Gromov-Witten invariants $K_{\mu,g,Q}$ can also be computed by
the topological recursion formula \cite{BKMP}. It is natural to ask
how to prove the Conjecture \ref{LMOVforopenstring} for the case of
toric Calabi-Yau 3-fold. Actually, this paper is devoted to this
conjecture for the resolved conifold with one AV-brane of framing
$\tau$.

\subsection{Lower genus cases}
We illustrate some lower genus cases for the above multiple covering
formula (\ref{openmultiplecovering}). By using the expansion $\sin
x=\sum_{k\geq 1}\frac{x^{2k-1}}{(2k-1)!}$, and taking the
coefficients of $g_s^{2g-2+l(\mu)}a^Q$ in formula
(\ref{openmultiplecovering}), we obtain
\begin{align} \label{multipecoveringgenus0}
K_{\mu,0,Q}=\sum_{d|\mu}(-1)^{l(\mu)}d^{l(\mu)-3}n_{\frac{\mu}{d},0,\frac{Q}{d}},
\end{align}
\begin{align*}
K_{\mu,1,Q}=\sum_{d|\mu}(-1)^{l(\mu)+1}\left(d^{l(\mu)-1}
n_{\frac{\mu}{d},1,\frac{Q}{d}}+\left(\frac{\sum_{j=1}^{l(\mu)}\mu_j^2}{24}d^{l(\mu)-3}-\frac{1}{12}d^{l(\mu)-1}\right)
n_{\frac{\mu}{d},0,\frac{Q}{d}}\right)
\end{align*}
\begin{align*}
&K_{\mu,2,Q}=\sum_{d|\mu}(-1)^{l(\mu)}\left(d^{l(\mu)+1}n_{\frac{\mu}{d},2,\frac{Q}{d}}
+\frac{\sum_{j=1}^{l(\mu)}\mu_j^2}{24}d^{l(\mu)-1}n_{\frac{\mu}{d},1,\frac{Q}{d}}\right.\\\nonumber
&\left.+\left(\frac{\sum_{j=1}^{l(\mu)}\mu_j^4}{1920}d^{l(\mu)-3}+\frac{\sum_{i<j}\mu_i^2\mu_j^2}{576}d^{l(\mu)-3}
-\frac{\sum_{j=1}^{l(\mu)}\mu_j^2}{288}d^{l(\mu)-1}+\frac{1}{240}d^{l(\mu)+1}\right)n_{\frac{\mu}{d},0,\frac{Q}{d}}\right)
\end{align*}
for $g=0$,  $g=1$ and $g=2$ respectively. These formulae were
firstly illustrated in \cite{MV}.

Therefore
\begin{align}  \label{multipecoveringlhole}
F_{(0,l)}&=\sum_{|\mu|=l}\sum_{Q}K_{\mu,0,Q}a^Qx_1^{\mu_1}\cdots
x_l^{\mu_l}\\\nonumber
&=\sum_{|\mu|=l}\sum_{Q}\sum_{d|\mu}(-1)^{l(\mu)}d^{l(\mu)-3}n_{\frac{\mu}{d},0,\frac{Q}{d}}a^{Q}x_1^{\mu_1}\cdots
x_l^{\mu_l}\\\nonumber & =(-1)^l\sum_{|\mu|=l}\sum_{Q}\sum_{d\geq
1}d^{l-3}n_{\mu,0,Q}a^{Q}x_1^{d\mu_1}\cdots x_l^{d\mu_l}.
\end{align}
In particular
\begin{align} \label{disccoutingformula}
F_{(0,1)}=-\sum_{m\geq 1}\sum_{d\geq
1}\sum_{Q}\frac{n_{m,0,Q}}{d^2}a^{dQ}x^{dm},
\end{align}
%
and for $g=1,l=1$,
\begin{align*}
F_{(1,1)}&=\sum_{m\geq 0}\sum_{Q}K_{(m),1,Q}a^{Q}x^{m}\\\nonumber
&=\sum_{m\geq
0}\sum_{Q}\left(\sum_{d|m}n_{m/d,1,Q/d}+(\frac{m^2}{24}d^{-2}-\frac{1}{12})n_{m/d,0,Q/d}\right)a^Qx^m\\\nonumber
&=\sum_{m\geq 0}\sum_{Q}\sum_{d\geq
1}\frac{1}{d}\left(n_{m,1,Q}+(\frac{m^2}{24}-\frac{1}{12})\right)a^{dQ}x^{dm}.
\end{align*}

\section{Chern-Simons theory and large $N$ duality}

\subsection{Quantum invariants}
In his seminal paper \cite{W1}, E. Witten introduced a new
topological invariant of a 3-manifold $M$  as a partition function
of quantum Chern-Simons theory. Let $G$ be a compact gauge group
which is a Lie group, and $M$ be an oriented three-dimensional
manifold. Let $\mathcal{A}$ be a $\mathfrak{g}$-valued connection on
$M$ where $\mathfrak{g}$ is the Lie algebra of $G$. The Chern-Simons
\cite{CS} action is given by
\begin{align*}
S(\mathcal{A})=\frac{k}{4\pi}\int_{M}Tr\left(\mathcal{A}\wedge
d\mathcal{A}+\frac{2}{3}\mathcal{A}\wedge\mathcal{A}
\wedge\mathcal{A}\right)
\end{align*}
where $k$ is an integer called the level.

Chern-Simons partition function is defined as the path integral in
quantum field theory
\begin{align*}
Z^G(M;k)=\int e^{i S(A)}D \mathcal{A}
\end{align*}
where the integral is over the space of all $\mathfrak{g}$-valued
connections $\mathcal{A}$ on $M$. Although it is not rigorous,
Witten \cite{W1} developed some techniques to calculate such
invariants.

If the three-manifold $M$ contains a link $\mathcal{L}$, we let
$\mathcal{L}$ be an $L$-component link denoted by
$\mathcal{L}=\bigsqcup _{j=1}^L\mathcal{K}_j$. Define
$$W_{R_j}(\mathcal{K}_j)=Tr_{R_j}Hol_{\mathcal{K}_j}(\mathcal{A})$$ which is the trace of holomony
along $\mathcal{K}_j$ taken in representation $R_j$. Then Witten's
invariant of the pair $(M,\mathcal{L})$ is given by
\begin{align*}
Z^{G}(M,\mathcal{L};\{R_j\};k)=\int e^{iS(\mathcal{A})}\prod_{j=1}^L
W_{R_j}(\mathcal{K}_j)D\mathcal{A}.
\end{align*}

When $M=S^3$ and the Lie algebra of $G$ is semisimple, Reshetikhin
and Turaev \cite{RT1,RT2} developed a systematic way to constructed
the above invariants by using the representation theory of quantum
groups. Their construction led to the definition of colored
HOMFLY-PT invariants  \cite{LM2,LZ}, which can be viewed as the
large $N$ limit of the quantum $U_{q}(sl_N)$ invariants. Usually, we
use the notation $W_{\lambda^1,..,\lambda^L}(\mathcal{L};q,a)$ to
denote the (framing-independent) colored HOMFLY-PT invariants for a
(oriented) link $\mathcal{L}=\bigsqcup _{j=1}^L\mathcal{K}_j$, where
each component $\mathcal{K}_j$ is colored by an irreducible
representation $V_{\lambda^j}$ of $U_{q}(sl_N)$. Some basic
structures for $W_{\lambda^1,..,\lambda^L}(\mathcal{L};q,a)$ were
proved in \cite{LP1,LP2,Zhu2}. It is difficult to obtain an explicit
formula of a given link for any irreducible representations
$\lambda$. We refer to \cite{LZ} for an explicit formula for torus
links, and a series of works due to Morozov et al \cite{MMM} and
Nawata et al \cite{NRZ} for some conjectural formulae of twist
knots. In particular, we have the following explicit formula for a
trivial knot (unknot) $U$:
\begin{align}\label{unknotformula}
W_{\lambda}(U;q,a)=\prod_{x\in
\lambda}\frac{a^{-1/2}q^{cn(x)/2}-a^{1/2}q^{-cn(x)/2}}{q^{h(x)/2}-q^{-h(x)/2}}.
\end{align}
For a box $x=(i,j)\in \lambda$, the hook length and content are
defined to be $hl(x)=\lambda_i+\lambda_j^{t}-i-j+1$ and $cn(x)=j-i$
respectively.

\subsection{Large $N$ duality} \label{sectionlargeN}
In another fundamental work of Witten \cite{W3}, the $SU(N)$
Chern-Simons gauge theory on a
three-manifold $M$ was interpreted as an open topological string theory on $%
T^*M$ with $N$ topological branes wrapping $M$ inside $T^*M$.
Furthermore, Gopakumar and Vafa \cite{GV2} conjectured that the
large $N$ limit of $SU(N)$ Chern-Simons gauge theory on $S^3$ is
equivalent to the closed topological string theory on the resolved
conifold. Furthermore, Ooguri and Vafa \cite{OV} generalized the
above construction to the case of a knot $\mathcal{K}$ in $S^3$.
They introduced the Chern-Simons partition function
$Z_{CS}^{(S^3,\mathcal{K})}(q,a,\mathbf{x})$ for $(S^3,\mathcal{K})$
which is a generating function of the colored HOMFLY-PT invariants
in all irreducible representations.
\begin{align} \label{chernsimonspartition}
Z_{CS}^{(S^3,\mathcal{K})}(q,a,\mathbf{x})=\sum_{\lambda\in
\mathcal{P}}W_{\lambda}(\mathcal{L},q,a)s_{\lambda}(\mathbf{x}).
\end{align}
Ooguri and Vafa \cite{OV} conjectured that for any knot
$\mathcal{K}$ in $S^3$, there exists a corresponding Lagrangian
submanifold $\mathcal{D}_{\mathcal{K}}$, such that the Chern-Simons
partition function
 is equal to the open topological string
partition function
 on
$(X,\mathcal{D}_{\mathcal{K}})$. They have established this duality
for the case of a trivial knot $U$ in $S^3$, and the link case was
further discussed in \cite{LMV}.

 In general,  we first
should find a way to construct the Lagrangian submanifold
$\mathcal{D}_\mathcal{L}$ corresponding to the link $\mathcal{L}$ in
geometry. See \cite{LMV,Koshkin,Tau,DSV} for the constructions for
some special links. Furthermore, if the Lagrangian submanifold
$\mathcal{D}_{\mathcal{L}}$ is constructed, then we need to compute
the open sting partition function under this geometry. For a trivial
knot in $S^3$, the dual open string partition function was computed
by J. Li and Y. Song \cite{LS} and S. Katz and C.-C.M. Liu
\cite{KL}.

On the other hand side, Aganagic and Vafa \cite{AV} introduced the
special Lagrangian submanifold in toric Calabi-Yau 3-fold which we
call Aganagic-Vafa A-brane (AV-brane) and studied its mirror
geometry, then they computed the counting of holomorphic disc end on
AV-brane by using the idea of mirror symmetry. Moreover, Aganagic
and Vafa surprisingly found the computation by using mirror symmetry
and the result from Chern-Simons knot invariants \cite{OV} are
matched. Furthermore, in \cite{AKV}, Aganagic, Klemm and Vafa
investigated the integer ambiguity appearing in the disc counting
and discovered that the corresponding ambiguity in Chern-Simons
theory was described by the framing of the knot. They checked that
the two ambiguities match for the case of the unknot, by comparing
the disk amplitudes on both sides.

Then, Mari\~no and Vafa \cite{MV} generalized the large $N$ duality
to the case of knots with arbitrary framing. They studied carefully
and established the large $N$ duality between a framed unknot in
$S^3$ and the open string theory on resolved conifold with AV-brane
by using the mathematical approach in \cite{KL}.  By comparing the
coefficient of the highest degree of the K\"{a}hler parameter in
this duality, they derived a remarkable Hodge integral identity
which now is called the Mari\~no-Vafa formula. Two mathematical
proofs for the Mari\~no-Vafa formula were given in \cite{LLZ} and
\cite{OP} respectively. We describe this duality in more details.
For a framed knot $\mathcal{K}_\tau$ with framing $\tau\in
\mathbb{Z}$, we define the framed colored HOMFLYPT invariants
$\mathcal{K}_\tau$ as follow,
\begin{align} \label{framedknotformula}
\mathcal{H}_\lambda(\mathcal{K}_\tau,q,a)=(-1)^{|\lambda|\tau}q^{\frac{\kappa_\lambda\tau}{2}}W_{\lambda}(\mathcal{K},q,a),
\end{align}
where
$\kappa_\lambda=\sum_{i=1}^{l(\lambda)}\lambda_i(\lambda_i-2i+1)$.

 The Chern-Simon partition function for
$(S^3,\mathcal{K}_\tau)$ is given by
\begin{align} \label{partitionfunctionframedunknot}
Z_{CS}^{(S^3,\mathcal{K}_\tau)}(q,a;\mathbf{x})=\sum_{\lambda\in
\mathcal{P}}\mathcal{H}_\lambda(\mathcal{K}_\tau,q,a)s_{\lambda}(\mathbf{x}).
\end{align}
We let $\hat{X}:=\mathcal{O}(-1)\oplus\mathcal{O}(-1)\rightarrow
\mathbb{P}^1$ be the resolved conifold, and $D_{\tau}$ be the
corresponding AV-brane. The open string partition function for
$(\hat{X},D_{\tau})$ has the structure
\begin{align} \label{partitonfunctionresolvedconifold}
Z_{str}^{(\hat{X},D_{\tau})}(g_s,a;\mathbf{x})=\exp\left(
-\sum_{g\geq0,\mu}\frac{\sqrt{-1}^{l(\mu)}}{|Aut(\mu)|}g_s^{2g-2+l(\mu)}F_{\mu,g}^{
\tau}(a)p_{\mu}(\mathbf{x})\right)
\end{align}
where $ F^{\tau}_{\mu,g}(a)=\sum_{Q\in
\mathbb{Z}/2}K^{\tau}_{\mu,g,Q}a^Q $ and $K_{\mu,g,Q}^{\tau}$ is the
open Gromov-Witten invariants
\begin{align*}
K_{\mu,g,Q}^{\tau}=\int_{[\mathcal{M}_{g,l(\mu)}(D^2,S^1|2Q,\mu_1,..,\mu_{l})]}e(\mathcal{V}),
\end{align*}
which is defined by S. Katz and C.-C. Liu \cite{KL}.  In particular,
when $Q=\frac{|\mu|}{2}$, the
 computations in \cite{KL} shows
\begin{align}\label{MVGW}
&K^{\tau}_{\mu,g,\frac{|\mu|}{2}}=(-1)^{|\mu|\tau}(\tau(\tau+1))^{l(\mu)-1}\\\nonumber
&\prod_{i=1}^{l(\mu)}
\frac{\prod_{j=1}^{\mu_i-1}(\mu_i\tau+j)}{(\mu_i-1)!}\int_{\overline{\mathcal{M}}_{g,l(\mu)}}
\frac{\Lambda_{g}^{\vee}(1)\Lambda_{g}^{\vee}(-\tau-1)\Lambda_g^{\vee}(\tau)}{\prod_{i=1}^{l(\mu)}(1-\mu_j\psi_j)}
\end{align}
where
$\Lambda_g^{\vee}(\tau)=\tau^g-\lambda_1\tau^{g-1}+\cdots+(-1)^g\lambda_g$.
Therefore, the large $N$ duality in this case is given by the
following identity:
\begin{align} \label{LargeNdualiyofunknot2}
Z_{CS}^{(S^3,U_\tau)}(q,a;\mathbf{x})=Z_{str}^{(\hat{X},D_{\tau})}(g_s,a;\mathbf{x})
\end{align}
where $q=e^{ig_s}$. By taking the coefficients of
$a^{\frac{|\mu|}{2}}$  of the following equality:
\begin{align*}
[p_{\mu}(\mathbf{x})g_s^{2g-2+l(\mu)}]\log
Z_{CS}^{(S^3,U_\tau)}(q,a;\mathbf{x})=[p_\mu(\mathbf{x})g_s^{2g-2+l(\mu)}]\log
Z_{str}^{(\hat{X},D_{\tau})}(g_s,a;\mathbf{x}),
\end{align*}
 we obtain the
Mari\~no-Vafa formula which is a Hodge integral identity with triple
$\lambda$-classes. The Mari\~no-Vafa formula provides a very
powerful tool to study the intersection theory of moduli space of
curves. From it, we can derive the Witten conjecture \cite{W2,Ko0},
 ELSV formula \cite{ELSV}, and various interesting Hodge integral identities,
see \cite{LLZ1,Li,Zhu1}.

Combining the idea of dualities shown above, and together with
several new technical ingredients,  Aganagic, Klemm, Mari\~no and
Vafa finally developed a systematic method, gluing the topological
vertices, to compute all loop topological string amplitudes on toric
Calabi-Yau manifolds \cite{AMV,AKMV}. The mathematical theory for
topological vertex was finally established in \cite{LLLZ}.  This
method provides an effective way to compute both the closed and open
string partition function for a toric Calabi-Yau 3-fold  with
AV-brane. Therefore, we have an explicit formula for the partition
function of resolved conifold
$Z_{str}^{(\hat{X},D_\tau)}(g_s,a,\mathbf{x})$, by comparing to the
explicit formula $Z_{CS}^{(S^3,U_\tau)}(q,a,\mathbf{x})$ of
Chern-Simons partition function describe above. Finally,  J. Zhou
proved the identity (\ref{LargeNdualiyofunknot2}) in \cite{Zhou}
based on the results in their previous works \cite{LLZ,LLZ2,LLLZ}.

\subsection{Integrality of the quantum invariants}
Now, let us collect the above discussions together. Let
$\mathcal{L}$ be a link in $S^3$, the large $N$ duality predicts
there exists a Lagrangian submanifold $\mathcal{D}_\mathcal{L}$ in
the resolved confold $\hat{X}$, and provides us the identity
(\ref{LargeNdualiyofunknot2}). Since the topological string
partition function
$Z_{str}^{(\hat{X},\mathcal{D}_\mathcal{L})}(g_s,a,\mathbf{x})$ has
the integrality structures by the discussions in Section 2.3, it
implies that the Chern-Simons partition function
$Z_{CS}^{(S^3,\mathcal{L})}(q,a,\mathbf{x})$ also inherits this
integrality structure. Usually, this integrality structure is called
the LMOV conjecture for link in \cite{LP1}. Furthermore, as
mentioned previously, the large $N$ duality was generalized to the
case of framed knot $\mathcal{K}_\tau$ \cite{MV}, where the
Chern-Simons partition $Z_{CS}^{(S^3,\mathcal{K}_\tau)}$ for framed
knot $\mathcal{K}_\tau$ is given by formula
(\ref{partitionfunctionframedunknot}). For convenience, we only
formulate the LMOV conjecture for framed knot $\mathcal{K}_\tau$ in
the following, although the conjecture should also holds for any
framed link, see \cite{LP3}.
\begin{conjecture}[LMOV conjecture for framed knot or framed LMOV conjecture]
\label{LMOVframedknot} Let
\begin{align*}
F_{CS}^{(S^3,\mathcal{K}_\tau)}(q,a,\mathbf{x})=\log
Z_{CS}^{(S^3,\mathcal{K}_\tau)}(q,a,\mathbf{x})
\end{align*}
be the Chern-Simons free energy for a framed knot $\mathcal{K}_\tau$
in $S^3$. Then there exist functions
$f_{\lambda}(\mathcal{K}_\tau;q,a)$ such that
\begin{align*}
F_{CS}^{(S^3,\mathcal{K}_\tau)}(q,a,\mathbf{x})=\sum_{d=1}^\infty\frac{1}{d}\sum_{\lambda\in
\mathcal{P},\lambda\neq
0}f_{\lambda}(\mathcal{K}_\tau;q^d,a^d)s_{\lambda}(\mathbf{x}^d).
\end{align*}
Let $
\hat{f}_{\mu}(\mathcal{K}_\tau;q,a)=\sum_{\lambda}f_{\lambda}(\mathcal{K}_\tau;q,a)M_{\lambda\mu}(q)^{-1},
$ where $M_{\lambda\mu}(q)$ is defined in the formula
(\ref{Mlambdamu}). Denote $z=q^{\frac{1}{2}}-q^{-\frac{1}{2}}$, then
for any $\mu\in \mathcal{P}^+$, there are integers
$N_{\mu,g,Q}(\tau)$ such that
\begin{align*}
\hat{f}_{\mu}(\mathcal{K}_\tau;q,a)=\sum_{g\geq
0}\sum_{Q}N_{\mu,g,Q}(\tau)z^{2g-2}a^Q\in
z^{-2}\mathbb{Z}[z^{2},a^{\pm \frac{1}{2}}].
\end{align*}
Therefore,
\begin{align*}
\mathfrak{z}_\mu\hat{g}_{\mu}(\mathcal{K}_\tau;q,a)&=\sum_{\nu}\chi_{\nu}(C_\mu)\hat{f}_\nu(\mathcal{K}_\tau;q,a)\\\nonumber
&=\sum_{g\geq 0}\sum_{Q}n_{\mu,g,Q}(\tau)z^{2g-2}a^Q\in
z^{-2}\mathbb{Z}[z^{2},a^{\pm \frac{1}{2}}].
\end{align*}
where
$n_{\mu,g,Q}(\tau)=\sum_{\nu}\chi_{\nu}(C_\mu)N_{\nu,g,Q}(\tau)$.
\end{conjecture}
 K. Liu and P. Peng \cite{LP1} first studied the
mathematical structures of LMOV conjecture for general links without
framing contribution (i.e. as to the Chern-Simons partition
(\ref{chernsimonspartition})), which is equivalent to the framed
LMOV conjecture for any links in framing zero. They provided a proof
for this case by using cut-and-join analysis and the cabling
technique \cite{LP1}. Motivated by the work \cite{MV}, K. Liu and P.
Peng \cite{LP3} formulated the framed LMOV conjecture for any
links(as to the Chern-Simons partition function
(\ref{partitionfunctionframedunknot}). In \cite{CLPZ}, the author
together with Q. Chen, K. Liu and P. Peng,  developed the ideas in
\cite{LP3} to study the mathematical structures of framed LMOV
conjecture and discovered the new structures named congruence skein
relations for colored HOMFLY-PT invariants.

\section{LMOV invariants for framed unknot $U_\tau$}
In Section \ref{sectionlargeN}, we have showed that, for a framed
unknot $U_\tau$ in $S^3$, the large $N$ duality holds \cite{Zhou}:
\begin{align*}
Z_{CS}^{(S^3,U_\tau)}(q,a;\mathbf{x})=Z_{str}^{(\hat{X},D_{\tau})}(g_s,a;\mathbf{x}),
\  q=e^{\sqrt{-1}g_s}.
\end{align*}
So one can compute LMOV invariants completely by using the colored
HOMFLY-PT invariants of the framed unknot $U_\tau$. On the other
hand side, by using mirror symmetry, one can also compute the
partition function $Z_{str}^{(\hat{X},D_{\tau})}(g_s,a;\mathbf{x})$
from B-model. The mirror geometry information of
$(\hat{X},D_{\tau})$ is encoded in a mirror curve
$\mathcal{C}_{\hat{X}}$. The disc counting information of
$(\hat{X},D_{\tau})$ is given by the superpotential related  to the
mirror curve \cite{AV,AKV}, and this fact was proved in \cite{FL}.

Furthermore, the open Gromov-Witten invariants of higher genus with
more holes can be computed by using Eynard-Orantin topological
recursions \cite{EO1}. This approach named as BKMP conjecture, was
proposed by Bouchard, Klemm, Mari\~no and Pasquetti \cite{BKMP}, and
was fully proved in \cite{EO2,FLZ} for any toric Calabi-Yau 3-fold
with AV-brane, so one can also use the BKMP method to compute the
LMOV invariants for $(\hat{X},D_{\tau})$. To determine the mirror
curve of $(\hat{X},D_{\tau})$, there are standard methods in toric
geometry. However, in \cite{AV2}, Aganagic and Vafa proposed another
effective way to compute the mirror curve, their method can be
applied to more general large $N$ geometry of an arbitrary knot in
$S^3$ \cite{AENV}. The rest contents of this section will be
organized as follow, we first illustrate the computations of the
mirror curve of $(\hat{X},D_{\tau})$ by using the method in
\cite{AV2}. Then, we compute the explicit formulae for genus 0 LMOV
invariants.  Next, we obtain the higher genus LMOV invariants with
one hole.

\subsection{a-deformed A-polynomial as the mirror curve}
The method used in \cite{AV2} to compute the mirror curve is based
on the fact that, colored HOMFLY-PT invariants colored by a
partition with a single row is a $q$-holonomic function, this fact
was conjectured and used in many literatures, such as
\cite{FGSA,FGS}, and was finally proved in \cite{GLL}. In fact, such
idea can go back to \cite{GL}.

Now, we illustrate such computations for framed unknot $U_\tau$. We
first compute the noncommutative a-deformed $A$-polynomial (it is
called the Q-deformed A-polynomial in \cite{AV2}, the variable $Q$
in \cite{AV2} is the variable $a$ here) for $U_\tau$.

By formula (\ref{unknotformula}), the colored HOMFLY-PT invariants
colored by partition $(n)$ for the unknot $U$ is given by
\begin{align*}
W_{n}(U;q,a)=\frac{a^{\frac{1}{2}}-a^{-\frac{1}{2}}}{q^{\frac{1}{2}}-q^{-\frac{1}{2}}}\cdots
\frac{a^{\frac{1}{2}}q^{\frac{n-1}{2}}-a^{-\frac{1}{2}}q^{\frac{-n-1}{2}}}{q^{\frac{n}{2}}-q^{-\frac{n}{2}}}
\end{align*}
It gives the recursion
\begin{align*}
(q^{n+1}-1)W_{n+1}(U;q,a)-(a^{\frac{1}{2}}q^{n+\frac{1}{2}}-a^{-\frac{1}{2}}q^{\frac{1}{2}})W_n(U;q,a)=0.
\end{align*}
By formula (\ref{framedknotformula}), the framed colored HOMFLY-PT
invariants for the framed unknot with framing $\tau\in \mathbb{Z}$
is
\begin{align*}
\mathcal{H}_n(U_\tau;q,a)=(-1)^{n\tau}q^{\frac{n(n-1)}{2}\tau}W_n(U;q,a).
\end{align*}
Then we obtain the recursion for $\mathcal{H}_n(U_\tau;q,a)$ as
follow
\begin{align} \label{frameunknotrecursion}
(-1)^\tau(q^{n+1}-1)\mathcal{H}_{n+1}(U_\tau;q,a)-(a^{\frac{1}{2}}q^{n+\frac{1}{2}}-a^{-\frac{1}{2}}q^{\frac{1}{2}})q^{n\tau}\mathcal{H}_n(U_\tau;q,a)=0.
\end{align}

For a general series $\{\mathcal{H}_n(q,a)\}_{n\geq 0}$,  we
introduce two operators $M$ and $L$ act on
$\{\mathcal{H}_n(q,a)\}_{n\geq 0}$ as follow:
\begin{align*}
M\mathcal{H}_n=q^{n}\mathcal{H}_n,\
L\mathcal{H}_{n}=\mathcal{H}_{n+1}.
\end{align*}
then $LM=qML$.
\begin{definition}
The noncommutative a-deformed A-polynomial for series
$\{\mathcal{H}_n(q,a)\}_{n\geq 0}$ is a polynomial
$\hat{A}(M,L;q,a)$ of operators $M, L$, such that
\begin{align*}
\hat{A}(M,L;q,a)\mathcal{H}_n(q,a)=0, \text{for } \ n\geq 0.
\end{align*}
and $A(M,L;a)=\lim_{q\rightarrow 1}\hat{A}(M,L;q,a)$ is called the
a-deformed A-polynomial.
\end{definition}
Therefore, from the recursion  (\ref{frameunknotrecursion}),  we
obtain the noncommutative a-deformed A-polynomial for $U_\tau$ as
follow:
\begin{align*}
\hat{A}_{U_\tau}(M,L,q;a)=(-1)^\tau(qM-1)L-M^{\tau}(a^{\frac{1}{2}}q^{\frac{1}{2}}M-a^{-\frac{1}{2}}q^{\frac{1}{2}}).
\end{align*}
And the a-deformed A-polynomial is given by
\begin{align*}
A_{U_\tau}(M,L;a)=\lim_{q\rightarrow
1}\hat{A}(M,L,q;a)=(-1)^\tau(M-1)L-M^{\tau}(a^{\frac{1}{2}}M-a^{-\frac{1}{2}}).
\end{align*}

In order to get the mirror curve of $U_\tau$, we need the following
general result which is written in the following lemma. Let
$Z(x)=\sum_{k\geq 0}\mathcal{H}_k(q,a)x^k$ be a generating function
of the series $\{\mathcal{H}_k(q,a)|k\geq 0\}$. We also introduce
two operators $\hat{x},\hat{y}$ act on $Z(x)$ as follow:
\begin{align*}
\hat{x}Z(x)=xZ(x), \ \hat{y}Z(x)=Z(qx).
\end{align*}
then $\hat{y}\hat{x}=q\hat{x}\hat{y}$. It is easy to obtain the
following result (see Lemma 2.1 in \cite{GKS} for the similar
statement).
\begin{proposition}
Given a noncommutative A-polynomial
$\hat{A}(M,L,q,a)=\sum_{i,j}c_{i,j}M^iL^j$ for the series
$\{\mathcal{H}_k(q,a)|k\geq 0\}$, then we have
\begin{align} \label{ApolynomialforZ}
\hat{A}(\hat{y},\hat{x}^{-1},q,a)Z(x)=\sum_{i,j}\sum_{-j\leq k\leq
-1}\mathcal{H}_{k+j}q^{ki}x^k.
\end{align}
\end{proposition}
\begin{proof}
Since
\begin{align*}
\hat{A}(\hat{y},\hat{x}^{-1},q,a)Z(x)&=\sum_{i,j}c_{i,j}\hat{y}^i\hat{x}^{-j}Z(x)
\\\nonumber
&=\sum_{i,j}c_{i,j}q^{-ij}x^{-j}Z(q^ix)\\\nonumber
&=\sum_{i,j}c_{i,j}\sum_{n\geq
0}\mathcal{H}_nq^{(n-j)i}x^{n-j}\\\nonumber
&=\sum_{i,j}c_{i,j}\sum_{k\geq 0}\mathcal{H}_{k+j}q^{k
i}x^k+\sum_{i,j}\sum_{-j\leq k\leq -1}a_{k+j}q^{ki}x^k.
\end{align*}
and by the definitions of the operators $M,L$,
$\hat{A}(M,L,q,a)\mathcal{H}_k=0$ gives
\begin{align*}
\sum_{i,j}c_{i,j}q^{ki}\mathcal{H}_{k+j}=0, \text{for} \ k\geq 0.
\end{align*}
We obtain the formula (\ref{ApolynomialforZ}).
\end{proof}
Finally, the mirror curve is given by
\begin{align*}
A(y,x^{-1};a)=\lim_{q\rightarrow
1}\hat{A}(\hat{y},\hat{x}^{-1};q,a)=0.
\end{align*}
Therefore, in our case, the mirror curve is:
\begin{align} \label{mirrorcurve2}
A_{U_\tau}(y,x^{-1};a)=y-1-a^{-\frac{1}{2}}(-1)^\tau
xy^{\tau}(ay-1)=0.
\end{align}

\subsection{Disc countings}
For convenience, we let $X=a^{-\frac{1}{2}}(-1)^\tau x$, and
$Y=1-y$, then the mirror curve (\ref{mirrorcurve2}) is changed to
the functional equation
\begin{align}  \label{mirrorcurvesimple}
Y=X(1-Y)^\tau(1-a(1-Y)).
\end{align}
In order to solve the above equation, we introduce the following
Lagrangian inversion formula \cite{Stanley}.
\begin{lemma}
Let $\phi(\lambda)$ be an invertible formal power series in the
indeterminate $\lambda$. Then the functional equation $Y=X\phi(Y)$
has a unique formal power series solution $Y=Y(X)$. Moreover, if $f$
is a formal power series, then
\begin{align} \label{lagrangeinversion}
f(Y(X))=f(0)+\sum_{n\geq
1}\frac{X^n}{n}\left[\frac{df(\lambda)}{d\lambda}\phi(\lambda)^{n}\right]_{\lambda^{n-1}}
\end{align}
\end{lemma}
\begin{remark}
In the following, we will frequently use the binomial coefficient
$\binom{n}{k}$ for all $n\in \mathbb{Z}$. That means for $n<0$, we
define $\binom{n}{k}=(-1)^k\binom{-n+k-1}{k}$.
\end{remark}

In our case, we take $ \phi(Y)=(1-Y)^\tau(1-a(1-Y)). $ Let
$f(Y)=1-Y$, by formula (\ref{lagrangeinversion}), we obtain
\begin{align*}
y(X)=1-Y(X)=1+\sum_{n\geq 1}\frac{X^n}{n}\sum_{i\geq
0}(-1)^{n+i}\binom{n}{i}\binom{n\tau +i}{n-1}a^i
\end{align*}
since $\phi(\lambda)^n$ has the expansion
\begin{align*}
\phi(\lambda)^n&=(1-\lambda)^{n\tau}(1-a(1-\lambda))^n\\\nonumber
&=\sum_{i\geq 0}\binom{n}{i}(-a)^i(1-\lambda)^{n\tau+i} \\\nonumber
&=\sum_{i,j\geq
0}\binom{n}{i}(-1)^{i+j}\binom{n\tau+i}{j}a^i\lambda^j.
\end{align*}
Moreover, if we let $f(Y(X))=\log(1-Y(X))$, then
\begin{align*}
\left[\frac{df(\lambda)}{d\lambda}\phi(\lambda)^{n}\right]_{\lambda^{n-1}}&=\sum_{i\geq
0}(-1)^{i}\binom{n}{i}\sum_{j=0}^{n-1}(-1)^{j+1}\binom{n\tau+i}{j}a^{i}\\\nonumber
&=\sum_{i\geq
0}(-1)^{i}\binom{n}{i}(-1)^{n}\binom{n\tau+i-1}{n-1}a^{i}
\end{align*}
where we have used the combinatoric identity:
\begin{align*}
\sum_{j=0}^{n-1}(-1)^{j+1}\binom{m}{j}=(-1)^n\binom{m-1}{n-1}.
\end{align*}
Formula (\ref{lagrangeinversion}) gives
\begin{align*}
\log(y(X))=\log(1-Y(X))=\sum_{n\geq 1}\frac{X^n}{n}\sum_{i\geq
0}(-1)^{n+i}\binom{n}{i}\binom{n\tau+i-1}{n-1}a^{i}.
\end{align*}
i.e.
\begin{align*}
\log(y(x))=\sum_{n\geq 1}\frac{x^n}{n}\sum_{i\geq
0}(-1)^{n\tau+n+i}\binom{n}{i}\binom{n\tau+i-1}{n-1}a^{i-\frac{n}{2}}.
\end{align*}
By BKMP's construction in genus 0 with one hole, one obtains
\begin{align} \label{disccoutingformulaunknot}
F_{(0,1)}&=\int \log(y(x))\frac{dx}{x}\\\nonumber &=\sum_{n\geq
1}\frac{x^n}{n^2}\sum_{i\geq
0}(-1)^{n\tau+n+i}\binom{n}{i}\binom{n\tau+i-1}{n-1}a^{i-\frac{n}{2}}.
\end{align}
By formula (\ref{disccoutingformula}), and if we let
$n_{m,l}(\tau)=n_{m,0,l-\frac{m}{2}}(\tau)$, then
\begin{align} \label{disccoutingformula2}
F_{(0,1)}=-\sum_{m\geq
1}\sum_{d|m,d|l}d^{-2}n_{\frac{m}{d},\frac{l}{d}}(\tau)x^ma^{l-\frac{m}{2}}.
\end{align}
Set
\begin{align*}
c_{m,l}(\tau)=-\frac{(-1)^{m\tau+m+l}}{m^2}\binom{m}{l}\binom{m\tau+l-1}{m-1},
\end{align*}
by comparing the coefficients of $x^ma^{l-\frac{m}{2}}$ in
(\ref{disccoutingformula2}) and (\ref{disccoutingformulaunknot}), we
have
\begin{align*}
c_{m,l}(\tau)=\sum_{d|m,d|l}\frac{n_{m/d,l/d}(\tau)}{d^2}.
\end{align*}
By M\"{o}bius inversion formula, we obtain
\begin{align} \label{formula-disccounting}
n_{m,l}(\tau)=\sum_{d|m,d|l}\frac{\mu(d)}{d^2}c_{\frac{m}{d},\frac{l}{d}}(\tau).
\end{align}
In \cite{LZhu}, we prove that $n_{m,l}(\tau)\in \mathbb{Z}$ by using
the basic method of number theory. Recently, Panfil and Sulkowski
\cite{PS} generalized the above disc counting formula
(\ref{formula-disccounting}) to a class of toric Calabi-Yau
manifolds without compact four-cycles which is also referred to as
strip geometries. In our notations, their formula (cf. formula
(4.19) in \cite{PS}) can be formulated as follow.

Given two integers $r,s\geq 0$, set $\mathbf{l}=(l_1,...,l_r)$,
$\mathbf{k}=(k_1,...,k_s)$, and $|\mathbf{l}|=\sum_{j=1}^rl_j$,
$|\mathbf{k}|=\sum_{j=1}^sk_j$. We define
\begin{align*}
c_{m,\mathbf{l},\mathbf{k}}(\tau)=\frac{(-1)^{m(\tau+1)+|\mathbf{l}|}}{m^2}
\binom{m\tau+|\mathbf{l}|+|\mathbf{k}|-1}{m-1}\prod_{j=1}^{r}\binom{m}{l_j}\prod_{j=1}^s
\frac{m}{m+k_j}\binom{m+k_j}{k_j}.
\end{align*}
Then, we have the disc counting formula
\begin{align} \label{formula-generaldiscounting}
n_{m,\mathbf{l},\mathbf{k}}(\tau)
=\sum_{d|\text{gcd}(m,\mathbf{l},\mathbf{k})}\frac{\mu(d)}{d^2}c_{m/d,\mathbf{l}/d,\mathbf{k}/d}(\tau)
\end{align}
It is obvious that formula (\ref{formula-disccounting}) is just the
special case of (\ref{formula-generaldiscounting}) by taking $r=1$
and $s=0$.

In \cite{Zhu5}, we will generalize the method used in \cite{LZhu} to
show that $n_{m,\mathbf{l},\mathbf{k}}(\tau)\in \mathbb{Z}$.

\subsection{Annulus counting}
The  Bergmann kernel of the curve (\ref{mirrorcurvesimple}) is
\begin{align*}
B(X_1,X_2)=\frac{dY_1dY_2}{(Y_1-Y_2)^2}.
\end{align*}
By the construction of BKMP \cite{BKMP}, the annulus amplitude is
calculated by the integral
\begin{align*}
\int
\left(B(X_1,X_2)-\frac{dX_1dX_2}{(X_1-X_2)^2}\right)=\ln\left(\frac{Y_2(X_2)-Y_1(X_1)}{X_2-X_1}\right)
\end{align*}

More precisely, for $m_1,m_2\geq 1$, the coefficients
$\left[\ln\left(\frac{Y_2(X_2)-Y_1(X_1)}{X_2-X_1}\right)\right]_{x_1^{m_1}x_2^{m_2}a^l}$
gives the annulus Gromov-Witten invariants $K_{(m_1,m_2),0,l}$.

Let
$b_{n,i}=\frac{(-1)^{n+i}}{n+1}\binom{n+1}{i}\binom{(n+1)\tau+i}{n}$
and $b_n=\sum_{i\geq 0}b_{n,i}a^i$. In particular $b_0=1-a$. Then
\begin{align*}
Y(X)=\sum_{n\geq 1}b_nX^n.
\end{align*}
and
\begin{align*}
\frac{Y_2(X_2)-Y_1(X_1)}{X_2-X_1}=(1-a)+\sum_{n\geq
1}b_n\left(\sum_{i=0}^{n}X_1^iX_2^{n-i}\right).
\end{align*}
Let $\tilde{b}_{m,l}=\sum_{i=0}^lb_{m,i}$  and
$\tilde{b}_{m}=\sum_{l=0}\tilde{b}_{m,l}a^l$. For $m_1\geq 1,
m_2\geq 1$, the coefficients $c_{(m_1,m_2)}$ of
$[X_1^{m_1}X_2^{m_2}]$ in the expansion
\begin{align*}
\ln\left(1+\sum_{n\geq
1}\tilde{b}_n\left(\sum_{i=0}^{n}X_1^iX_2^{n-i}\right)\right)
\end{align*}
is given by
\begin{align*}
c_{(m_1,m_2)}=\sum_{|\mu|=m_1+m_2}\frac{(-1)^{l(\mu)-1}(l(\mu)-1)!\tilde{b}_\mu}{|Aut(\mu)|}|S_{\mu}(m_1)|
\end{align*}
where $S_{\mu}(m_1)$ is the set
\begin{align*}
S_\mu(m_1)=\{(i_1,...,i_{l(\mu)})\in
\mathbb{Z}^{l(\mu)}|\sum_{k=1}^{l(\mu)}i_k=m_1,\ \text{where}\ 0\leq
i_k\leq \mu_k, \ \text{for} \ k=1,..,l(\mu)\},
\end{align*}
by this definition,  $S_\mu(m_1)=S_{\mu}(m_2)$.

We write $c_{(m_1,m_2)}=\sum_{l\geq 0}c_{(m_1,m_2),l}a^l$, then the
annulus amplitude is
\begin{align*}
F_{(0,2)}=\sum_{m_1\geq 1,m_2\geq 1}\sum_{l\geq
0}(-1)^{(m_1+m_2)\tau}c_{(m_1,m_2),l}a^{l-\frac{m_1+m_2}{2}}x_1^{m_1}x_{2}^{m_2}.
\end{align*}

Set $n_{(m_1,m_2),l}=n_{(m_1,m_2),0,l-\frac{m_1+m_2}{2}}$, the
multiple covering formula (\ref{multipecoveringlhole}) for $l=2$
gives
\begin{align*}
F_{(0,2)}=\sum_{m_1\geq 1,m_2\geq 1}\sum_{l\geq
0}\sum_{d|m_1,d|m_2,d|l}\frac{1}{d}n_{(\frac{m_1}{d},\frac{m_2}{d}),\frac{l}{d}}a^{l-\frac{m_1+m_2}{2}}x_1^{m_1}x_{2}^{m_2}
\end{align*}
we have
\begin{align*}
(-1)^{(m_1+m_2)\tau}c_{(m_1,m_2),l}=\sum_{d|m_1,d|m_2,d|l}\frac{1}{d}n_{(\frac{m_1}{d},\frac{m_2}{d}),\frac{l}{d}}.
\end{align*}
so
\begin{align*}
n_{(m_1,m_2),l}=\sum_{d|m_1,d|m_2,d|l}\frac{\mu(d)}{d}(-1)^{\frac{(m_1+m_2)\tau}{d}}c_{(\frac{m_1}{d},\frac{m_2}{d}),\frac{l}{d}}.
\end{align*}
In particular, when $l=\frac{m_1+m_2}{2}$, we only need to consider
the the curve $Y=X(1-Y)^\tau$.  With the help of the following
formula proved in \cite{Zhu1}
\begin{lemma} [Lemma 2.3 of \cite{Zhu1}]
\begin{align}
\ln\left(\frac{Y_{1}(X_1)-Y_{2}(X_2)}{X_{1}-X_{2}}\right)&=\sum_{m_1,m_2
\geq
1}\frac{1}{m_1+m_2}\binom{m_1\tau+m_1-1}{m_1}\binom{m_2\tau+m_2}{m_2}X_1^{m_1}X_2^{m_2}\\\nonumber
&-\tau\left(\ln(1-Y_{1}(X_1))+\ln(1-Y_{2}(X_2))\right).
\end{align}
\end{lemma}
We obtain
\begin{align} \label{annuluscountingunknot}
c_{(m_1,m_2),\frac{m_1+m_2}{2}}(\tau)&=\frac{1}{m_1+m_2}\binom{m_1\tau+m_1-1}{m_1}\binom{m_2\tau+m_2}{m_2}.
\end{align}

For brevity, we let
$n_{(m_1,m_2)}(\tau):=n_{(m_1,m_2),\frac{m_1+m_2}{2}}(\tau)$ which
is defined through formula (\ref{annuluscountingunknot}). Then we
obtain
\begin{align*}
    n_{(m_1,m_2)}(\tau)& =  \frac{1}{m_1+m_2}\sum_{d\mid m_1,d\mid m_2} \mu(d) (-1)^{(m_1+m_2)(\tau+1)/d} \nonumber \\
                     & \cdot
                     \binom{(m_1\tau+m_1)/d-1}{m_1/d}\binom{(m_2\tau+m_2)/d}{m_2/d}.
\end{align*}
In \cite{LZhu}, we have also proved that $n_{(m_1,m_2)}(\tau)\in
\mathbb{Z}$ for any $m_1,m_2\geq 1$ and $\tau\in \mathbb{Z}$.

\subsection{Genus $g=0$ with more  holes}
By formula (\ref{MVGW}), we have
\begin{align*}
K^\tau_{\mu,g,\frac{|\mu|}{2}}&=(-1)^{|\mu|\tau}[\tau(\tau+1)]^{l(\mu)-1}\prod_{i=1}^{l(\mu)}
\frac{\prod_{a=1}^{\mu_{i}-1}(\mu_{i}\tau+a)}{(\mu_{i}-1)!}\int_{\overline{\mathcal{M}}_{g,l(\mu)}}
\frac{\Gamma_{g}(\tau)}{\prod_{i=1}^{l(\mu)}(1-\mu_{i}\psi_{i})}\\\nonumber
&=(-1)^{|\mu|\tau}[\tau(\tau+1)]^{l(\mu)-1}\prod_{i=1}^{l(\mu)}\binom{\mu_i(\tau+1)-1}{\mu_i-1}\sum_{b_i\geq
0}\prod_{i=1}^{l(\mu)}\mu_i^{b_i}\langle\prod_{i=1}^{l(\mu)}\tau_{b_i}\Gamma_g(\tau)
\rangle_{g,l(\mu)}
\end{align*}
When $g=0$ and $l\geq 3$, then $\Gamma_{0}(\tau)=1$ and we have the
Hodge integral identity:
\begin{align*}
\langle\tau_{b_1}\cdots\tau_{b_l}\rangle_{0,l}=\binom{l-3}{b_1,..,b_l}.
\end{align*}
Hence, we obtain
\begin{align} \label{MVGWgenus0}
K^\tau_{\mu,0,\frac{|\mu|}{2}}=(-1)^{|\mu|\tau}[\tau(\tau+1)]^{l(\mu)-1}\prod_{i=1}^{l(\mu)}
\binom{\mu_i(\tau+1)-1}{\mu_i-1}\left(\sum_{i=1}^{l(\mu)}\mu_i\right)^{l(\mu)-3}.
\end{align}
By using formula (\ref{multipecoveringgenus0}), we get
\begin{align}
n_{\mu,0,\frac{|\mu|}{2}}(\tau)=(-1)^{l(\mu)}\sum_{d|\mu}\mu(d)d^{l(\mu)-1}K^{\tau}_{\frac{\mu}{d},0,\frac{|\mu|}{2d}}.
\end{align}
It is clear that $K^\tau_{\mu,0,\frac{|\mu|}{2}}\in \mathbb{Z}$ from
formula (\ref{MVGWgenus0}), and since $l(\mu)\geq 3$, it is clear
that
\begin{align*}
n_{\mu,0,\frac{|\mu|}{2}}(\tau)\in \mathbb{Z}
\end{align*}
for any partition $\mu$ with $l(\mu)\geq 3$.

\subsection{Genus $g \geq 1$ with one hole}
As discussed in the introduction, we have three approaches to the
LMOV invariants $n_{\mu,g,Q}(\tau)$ for the open topological string
model $(\hat{X},\mathcal{D}_\tau)$, as to the higher genus LMOV
invariants $n_{\mu,g,Q}(\tau)$, we would like to use the identity
(\ref{LargeNdualiyofunknot}) of large $N$ duality to change all the
computations from topological string to Chern-Simons theory for knot
invariants. Since the large $N$ duality of topological string and
Chern-Simons theory was conjectured for any framed knots (even
links), we formulate the LMOV integrality conjecture for any framed
knots first, and then we focus on the special case of framed unknot
$U_\tau$ in $S^3$, and give an explicit formula for the LMOV
invariants $n_{(m),g,Q}(\tau)$ whose integrality properties was
proved in \cite{LZhu}.

\subsubsection{Revist LMOV integrality conjecture for framed knot $\mathcal{K}_\tau$}
We introduce the following notations first. Let $n\in \mathbb{Z}$
and $\lambda,\mu,\nu$ denote the partitions. Set
\begin{align} \label{quantuminteger}
\{n\}_x=x^{\frac{n}{2}}-x^{-\frac{n}{2}}, \
\{\mu\}_{x}=\prod_{i=1}^{l(\mu)}\{\mu_i\}_x.
\end{align}
For brevity, we denote $\{n\}=\{n\}_q$ and $\{\mu\}=\{\mu\}_q$. Let
$\mathcal{K}_{\tau}$ be a knot with framing $\tau \in \mathbb{Z}$.
The framed colored HOMFLYPT invariant
$\mathcal{H}(\mathcal{K}_\tau;q,a)$ of $\mathcal{K}_\tau$ is defined
by formula (\ref{framedknotformula}). Let
\begin{align*}
\mathcal{Z}_{\mu}(\mathcal{K}_\tau)=\sum_{\lambda}\chi_{\lambda}(C_\mu)\mathcal{H}_{\lambda}(\mathcal{K}_\tau),
\end{align*}
then the Chern-Simons partition function is given by
\begin{align*}
Z_{CS}^{(S^3,\mathcal{K}_\tau)}=\sum_{\lambda\in
\mathcal{P}}\mathcal{H}_\lambda(\mathcal{K}_\tau)s_{\lambda}(x)
=\sum_{\mu\in
\mathcal{P}}\frac{\mathcal{Z}_{\mu}(\mathcal{K}_\tau)}{\mathfrak{z}_\mu}p_{\mu}(x).
\end{align*}
We define $F_{\mu}(\mathcal{K}_\tau)$ though the expansion formula
\begin{align*}
F_{CS}^{(S^3,\mathcal{K}_\tau)}=\log(Z_{CS}^{(S^3,\mathcal{K}_\tau)})=\sum_{\mu\in
\mathcal{P}^+}F_{\mu}(\mathcal{K}_\tau)p_{\mu}(x),
\end{align*}
then we have
\begin{align*}
F_{\mu}(\mathcal{K}_\tau)=\sum_{n\geq
1}\sum_{\cup_{i=1}^{n}\nu^i=\mu}\frac{(-1)^{n-1}}{n}\prod_{i=1}^{n}\frac{\mathcal{Z}_{\nu^i}(\mathcal{K}_\tau)}{\mathfrak{z}_{\nu^i}}.
\end{align*}

\begin{remark}
For two partitions $\nu^1$ and $\nu^2$, the notation $\nu^1\cup
\nu^2$ denotes the new partition by combing all the parts in $\nu^1,
\nu^2$. For example $\mu=(2,2,1)$, then the set of pairs
$(\nu^1,\nu^2)$ such that $\nu^1 \cup \nu^2 =(2,2,1)$ is
\begin{align*}
(\nu^1=(2), \nu^2=(2,1)),\ (\nu^1=(2,1), \nu^2=(2)),\\\nonumber
 \
(\nu^1=(1), \nu^2=(2,2)), \ (\nu^1=(2,2), \nu^2=(1)), \
\end{align*}
\end{remark}
For a rational function $f(q,a)\in \mathbb{Q}(q^{\pm},a^{\pm})$, we
define the adams operator
\begin{align*}
\Psi_{d}(f(q,a))=f(q^d,a^d).
\end{align*}
Then, we set
\begin{align} \label{formulagmu}
\hat{g}_{\mu}(\mathcal{K}_\tau)=\sum_{d|\mu}\frac{\mu(d)}{d}\Psi_{d}(\hat{F}_{\mu/d}(\mathcal{K}_\tau)),
\end{align}
where
\begin{align*}
\hat{F}_{\mu}(\mathcal{K}_\tau)=\frac{F_{\mu}(\mathcal{K}_\tau)}{\{\mu\}}.
\end{align*}

The LMOV integrality conjecture for framed knot $\mathcal{K}_\tau$
states that
\begin{conjecture} \label{LMOVframedknot2}
$\mathfrak{z}_\mu\hat{g}_{\mu}(\mathcal{K}_\tau)$ is a polynomial of
the LMOV invariants $n_{\mu,g,Q}(\tau)$, more precisely,
\begin{align*}
\mathfrak{z}_\mu\hat{g}_{\mu}(\mathcal{K}_\tau)=\sum_{g\geq
0}\sum_{Q}n_{\mu,g,Q}(\tau)z^{2g-2}a^Q\in
z^{-2}\mathbb{Z}[z^2,a^{\pm \frac{1}{2}}],
\end{align*}
where $z=q^{\frac{1}{2}}-q^{-\frac{1}{2}}=\{1\}$.
\end{conjecture}

\subsubsection{LMOV integrality invariants for $U_\tau$}
Now we apply the above computations to the case of framed unknot
$U_{\tau}$. By the large $N$ duality formula in this special case
(\ref{LargeNdualiyofunknot}) proved by \cite{Zhou}, the LMOV
integrality Conjecture \ref{LMOVforopenstring} for the open string
model $(\hat{X},\mathcal{D}_\tau)$ and LMOV integrality conjecture
for the Chern-Simon theory $(S^3,U_\tau)$ are same. In other words,
LMOV invariants for the $(\hat{X},\mathcal{D}_\tau)$ and
$(S^3,U_\tau)$ are the same one.

For convenience, we define the function
\begin{align*}
\phi_{\mu,\nu}(x)=\sum_{\lambda}\chi_{\lambda}(C_\mu)\chi_{\lambda}(C_\nu)x^{\kappa_\lambda}.
\end{align*}
By Lemma 5.1  in \cite{CLPZ},  for $d\in \mathbb{Z}_+$, we have
\begin{align*}
\phi_{(d),\nu}(x)=\frac{\{d\nu\}_{x^2}}{\{d\}_{x^2}}.
\end{align*}
By using the formula of colored HOMFLYPT invariant for unknot
(\ref{unknotformula}), we obtain
\begin{align*}
\mathcal{Z}_{\mu}(U_\tau)&=\sum_{\lambda}\chi_{\lambda}(C_\mu)\mathcal{H}_{\lambda}(U_\tau)\\\nonumber
&=(-1)^{|\mu|\tau}\sum_{\lambda}\chi_{\lambda}(C_\mu)q^{\frac{\kappa_\lambda
\tau}{2}}\sum_{\nu}\frac{\chi_{\lambda}(C_\nu)}{z_\nu}\frac{\{\nu\}_a}{\{\nu\}}\\\nonumber
&=(-1)^{|\mu|\tau}\sum_{\nu}\frac{1}{\mathfrak{z}_\nu}\phi_{\mu,\nu}(q^\frac{\tau}{2})\frac{\{\nu\}_a}{\{\nu\}}.
\end{align*}
In particular, for $\mu=(m)$, $m\in \mathbb{Z}$, we have
\begin{align*}
\mathcal{Z}_{m}(U_\tau)=(-1)^{m\tau}\sum_{|\nu|=m}\frac{1}{\mathfrak{z}_\nu}\frac{\{m\nu\tau\}}{\{m\tau\}}\frac{\{\nu\}_a}{\{\nu\}}.
\end{align*}
For brevity, we set
$\mathcal{Z}_m(q,a)=\frac{1}{\{m\}}\mathcal{Z}_{m}(U_\tau)=(-1)^{m\tau}\sum_{|\nu|=m}\frac{1}{\mathfrak{z}_\nu}
\frac{\{m\nu\tau\}}{\{m\}\{m\tau\}}\frac{\{\nu\}_a}{\{\nu\}}$ and
$g_m(q,a)=\mathfrak{z}_{(m)}\hat{g}_m(U_\tau)$. Then, by formula
(\ref{formulagmu}), we obtain
\begin{align} \label{gm}
g_m(q,a)=\sum_{d|m}\mu(d)\mathcal{Z}_{m/d}(q^d,a^d).
\end{align}

In \cite{LZhu}, we have proved that $g_{m}(q,a)$ is a polynomial of
 the higher genus with one hole LMOV invariants  $n_{m,g,Q}(\tau)$.
More precisely,
\begin{align*}
g_m(q,a)=\sum_{g\geq 0}\sum_{Q}n_{m,g,Q}(\tau)z^{2g-2}a^Q\in
z^{-2}\mathbb{Z}[z^2,a^{\pm \frac{1}{2}}],
\end{align*}
where $z=q^{\frac{1}{2}}-q^{-\frac{1}{2}}=\{1\}$. In other words, we
have
\begin{align}
n_{m,g,Q}(\tau)=\text{Coefficient of term $z^{2g-2}a^Q$ in the
polynomial } g_m(q,a).
\end{align}

\section{Conclusions and related works}
In this finial section, we mention some related works and problems
which are deserved to study further.

\begin{itemize}
\item Applications of our explicit formulae. In A. Mironov et
al's work \cite{MMMS}, they made a lot of numerical computations for
a large variety of LMOV invariants by using their recent works on
knot invariants. By experimental observation, they proposed a
conjecture that the absolute values of the LMOV invariants
$N_{\mu,g,Q}^{\mathcal{K}}$ for big enough representations $\mu$
approach Gaussian/binomial distribution in $g$ with just three
$\mu$, $Q$-dependent parameters (cf. Conjecture in Section 5 of
\cite{MMMS} and note that the meanings of the symbols $\mu,g, Q$
used here are corresponding to $Q,g,n$ in \cite{MMMS}). The LMOV
invariant $n_{\mu,g,Q}$ are related to $N_{\mu,g,Q}$ through a
character transformation formula (\ref{formula-invariantchange}). So
we expect that our explicit formulae could provide a rigid proof of
their conjecture at least for the case of framed unknot $U_\tau$.

\item Interpretations of the integrality of LMOV invariants.  In
\cite{LZhu0}, we found a relation between the open string partition
of $\mathbb{C}^3$ with AV brane $D_\tau$ and the Hilbert-Poincar\'e
series of the Cohomological Hall algebra of the $|\tau|$-loop quiver
in the sense of \cite{KS}. This is the first example of toric
Calabi-Yau and quiver correspondence. Then in
\cite{KRSS1,KRSS2,Zhu4}, a general knot-quiver correspondence was
proposed. Especially, Sulkowski et al. \cite{KRSS2,PSS} established
this correspondence for a large class of knot, and links.

\item Compute the explicit formulae for LMOV invariants in more general
settings. In the recent work of Panfil and Sulkowski \cite{PS}, they
found  a direct relation between quiver representation theory and
open topological string theory on a class of toric Calabi-Yau
3-folds referred as strip geometries. With the help of the relation
to quivers they also derive explicit expressions for classical open
BPS invariants for an arbitrary strip geometry, which lead to a
large set of number theoretic integrality statements. In
\cite{Zhu5}, we generalize our current work to the case of torus
knots and other settings, more explicit formulae for corresponding
LMOV invariants are obtained. We will develop a general number
theory method to prove integrality properties of LMOV invariants.

\end{itemize}

\end{document}